\documentclass{lipics-v2019}

\title{A Term-Rewriting Semantics for Imperative Style Programming}

\author{David A. Plaisted}{Department of Computer Science, UNC Chapel Hill, Chapel Hill, NC 27599-3175, U.S.A.,Phone: (919) 590-6051}{plaisted@cs.unc.edu}{}{}

\author{Lee Barnett}{Formal Models and Verification,Johannes-Kepler Universit\"atAltenbergerstraße 69, 4040 Linz, Austria}{lee.barnett@jku.at}{}{}

\authorrunning{D.\,A. Plaisted and L. Barnett}

\Copyright{David Plaisted and Lee Barnett}



\EventEditors{John Q. Open and Joan R. Access}
\EventNoEds{2}
\EventLongTitle{42nd Conference on Very Important Topics (CVIT 2016)}
\EventShortTitle{CVIT 2016}
\EventAcronym{CVIT}
\EventYear{2016}
\EventDate{December 24--27, 2016}
\EventLocation{Little Whinging, United Kingdom}
\EventLogo{}
\SeriesVolume{42}
\ArticleNo{23}

\begin{document}
\maketitle
\begin{abstract}
Term rewriting systems have a simple syntax and semantics and facilitate proofs of correctness.  However, they are not as popular in industry or academia as imperative languages.  We define a term rewriting based abstract programming language with an imperative style and a precise semantics allowing programs to be translatable into efficient imperative languages, to obtain proofs of correctness together with efficient execution.  This language is designed to facilitate translations into correct programs in imperative languages with assignment statements, iteration, recursion, arrays, pointers, and side effects.  It can also be used
in place of a pseudo-programming language to specify algorithms.
\end{abstract}

\section{Introduction}
Term rewriting systems\cite{bani:98} have a simple syntax and semantics.  A first-order term-rewriting based programming style is introduced that facilitates translations into imperative languages.
The purpose of this system is to provide a way to express abstract algorithms that has a simple syntax and semantics but is also close to imperative languages in style.  Then
abstract algorithms can be written and verified in such a system and translated into many imperative programming languages, so that the translated programs are correct and efficient.  This facilitates verification and also avoids the need to rewrite the same algorithm over and over again in many languages.  At this stage only pure algorithms are considered, and features such as interrupts and input-output are not considered.  After the abstract algorithm has been translated into a target language and inserted into a larger program, such features can be added in a language-specific way.  The idea of an abstract language that can be translated into others was presented earlier \cite{DBLP:journals/corr/cs-SE-0306028,Plaisted:2013}, but the abstract languages considered before were imperative and not fully specified or not specified at all.   It is {\em not} the purpose of this system to provide
the most efficient translation of term-rewriting systems into imperative languages, but
to provide an abstract term-rewriting based notation for algorithms with a precise, accessible syntax and semantics.

The abstract programs are written in a formalism that is pure and close to logical notation, so proofs of correctness may be easier than for imperative programs.  This is due to the simple syntax and semantics of first-order term-rewriting systems.

This system is untyped.  It is much simpler than other systems that have been proposed to give abstract descriptions of algorithms, such as Coq\cite{CoquandHuet:1988,Forster2018CallbyValueLC} or other logical frameworks\cite{gome:93,paulson:94}, which may use the Curry-Howard isomorphism\cite{Howard80}.  Thus this system may be easier for programmers to understand, though the ideas presented here can also be extended to more sophisticated formalisms.  It is {\em not} the purpose of this language to have the most expressive
features such as higher-order functions, nonlinear rules, overlapping rules or AC unification; the purpose is to keep the language as simple as possible while meeting its objectives.

The language is a combination of rewriting and single assignment style programming.  Side effects complicate the term rewriting semantics, but the semantics is still precisely defined.

\section{Syntax}

The usual definitions of terms, rewrite rules, and so on will be assumed \cite{bani:98}.  We concentrate on constructor systems that are
complete in that all terms of the form $f(t_1, \dots, t_n)$ for defined symbol $f$ and constructor terms $t_i$ are reducible.  Also,
we will use left-linear and non-overlapping (orthogonal) systems because they seem to correspond naturally to imperative programs.
Left linearity is a reasonable restriction for a programming language in the style of an imperative language,  because it corresponds to the lack of repeated formal parameters in imperative programming languages.  Disjointness is also reasonable because it means that each term can be rewritten in at most one way.

A {\em procedure} for non-constructor $f$  is a set of rewrite rules with all left-hand sides of the form $f(t_1, \dots, t_n)$ for some terms $t_i$.  If such a procedure consists of all rules in $R$ of the specified form then it is called the $R$-procedure for $f$.  A system $R$ is {\em complete} if for all non-constructors $f$ appearing in $R$, for all terms $f(t_1, \dots, t_n)$ where the $t_i$ are ground constructor terms, $f(t_1, \dots, t_n)$ is reducible by the $R$-procedure for $f$.

Completeness corresponds to the fact that all inputs can be processed by typical imperative languages.  It implies that all non-constructor ground terms are reducible; as a consequence, all normal forms of ground terms are constructor terms.
\begin{definition}
A {\em COC} system is a constructor term rewriting system that is complete and orthogonal.
\end{definition}
Such systems seem to correspond naturally to imperative programs.  In such a system $R$, if $R$ is terminating, then for each ground term $r$ there is a unique term $s$ such that $r \Rightarrow!_R s$, and $s$ must be a constructor term.  Also, term rewriting systems and constructor terms automatically give records and references.

A possible extension to the current formalism would be to allow nondeterminism.  Also, if $R$ is not terminating, then one can naturally extend the semantics to allow infinite constructor terms as normal forms, but this will not be considered here.

\subsection{Procedures}

We assume rewrite rules for $\wedge$, $\vee$, $\neg$, and other Boolean operators are given.  Binary Boolean operators will be used in infix notation.  Also, truth values $true$ and $false$ are constructor constants.  We also assume standard rules for (if  then else).

Here are some additional procedures that will be convenient

Top:

\begin{tabular}{l}
top($f(x_1, \dots, x_n)) \rightarrow c_f$\\
top($c$) $\rightarrow$ $c$ for constructor constants $c$\\
\end{tabular}

where $f$ is a constructor, $c_f$ is an individual constructor constant and $c_f$ is distinct for each $f$.

Equal:

\begin{tabular}{l}
$eq(c,c) \rightarrow$ true for constructor constants $c$\\
$eq(c,d) \rightarrow$ false for distinct constructor constants $c,d$\\
$eq(f(x_1, \dots, x_m), g(y_1, \dots, y_n)) \rightarrow$ $m = n \wedge$ $eq(top(f(x_1, \dots, x_m)), top(g(y_1, \dots, y_n)))$\\
$\wedge eq(x_1, y_1) \wedge \dots \wedge eq(x_m, y_m)$, one such rule needed for each pair $f$ and $g$ of constructors,  at\\
least one not a constant\\
\end{tabular}

Arg:

\begin{tabular}{l}
$arg(i,f(x_1, \dots, x_n)) \rightarrow x_i$\\
One rule needed for each different $i$, possibly infinitely many rules in all\\
\end{tabular}

Replace:

\begin{tabular}{l}
$replace(i,y,f(x_1, \dots, x_n)) \rightarrow f(x_1, \dots, y, \dots, x_n)$\\
where the $i$-th argument of $f$ has been replaced by $y$.\\
One rule is needed for each non-constant constructor $f$ and each $i$.\\
\end{tabular}

D-Replace:

There are also similar rules for $d\_replace(i,y,f(x_1, \dots, x_n))$ that return the same value but are destructive, in that the term
$f(x_1, \dots, x_n)$ is not copied but the $i^{th}$ argument is replaced.  This will be explained further in the following sections.

Tupling:

The {\em tupling} operator $\langle \dots \rangle$, a constructor, will also be needed, but there are no rules for it.  

Projection:

The Arg operator is specialized to $\pi_i$ for tuples:

$\pi_i \langle x_1, x_2, \dots, x_n \rangle = x_i$

Copy:

$copy(f(s_1, \dots, s_n)) \rightarrow f(s_1, \dots, s_n)$

This rule creates a new copy
of a term so that if the original term is destructively modified then the new term is not affected.

\subsection{Compiled Procedures}

Some procedures are {\em compiled}: If a procedure for $f$ is compiled, then when rewriting a term of the form $f(s_1, \dots, s_n)$ for constructor terms $s_i$, this term is replaced by a term $t$ generated by a program in some other language, possibly compiled. Examples include arithmetic operations in which numbers are considered as constants in the logic;  then for example $sum(3,5)$ would evaluate to $8$. Such a procedure can be represented by the set of rewrite rules $f(s_1, \dots, s_n) \rightarrow t$ of the above form, for enough tuples of constructor terms $s_i$ to guarantee completeness.  This could be an infinite system.  Procedures that are not compiled are represented by a finite term rewriting system.

If one represents the semantics of a compiled function $f$ as a function $f^*$ mapping tuples of
constructor terms to constructor terms, then the rewrite system for $f$ could be chosen to be $\{f(s_1, \dots, s_n) \rightarrow f^*(s_1, \dots, s_n): \mbox{the $s_i$ are ground constructor terms}\}$.

\subsection{Programs}
A {\em program} is a finite set of procedures.  The letters $P, Q$ will denote programs.  Procedures are indicated by the letters $p$ and $q$ and can either be
{\em rewrite procedures} or {\em flat} procedures.  Rewrite procedures $p$ are defined by 
a set $R_p$ of rewrite rules, with a possibly specified rewrite strategy for each procedure, and flat procedures $p$ are defined by
a set $R_p$ consisting of single rewrite rule of the form $p(x_1, \dots, x_n) \rightarrow
E_p[\overline{x}]$, where $E_p[\overline{x}]$ is as defined below.  There may also
be compiled procedures.  If $P$ is a program then
there is a term-rewriting system $R_P$ associated with $P$ which is $\bigcup \{R_p :
p \in P\}$.
If $R$ is a term rewriting system, then $R^=$ is the set of equations $r = s$ for all rules $r \rightarrow s$ in $R$.  If $P$ is a program then the {\em declarative semantics} ${\cal D}(P)$ of $P$ is $R_P^=$.  It is easily seen that if $s \Rightarrow_R^* t$ then ${\cal D}(P) \models (s = t)$.  The following result is well known.

\begin{theorem}
If $P$ is a $COC$ program, then there are no two distinct constructor terms $s$ and $t$ such that ${\cal D}(P) \models s = t$.
\end{theorem}

\begin{proof}
This follows from confluence of $COC$ systems and from the fact that all constructor terms are irreducible, using standard term-rewriting theory.
\end{proof}

\section{Example Procedures}
We now give some example procedures in this formalism.  For these, if the input is not of the expected term, then the result will be a constructor constant ``error" indicating that an error has occurred.

Here is the essential part of a (flat) binary search procedure, omitting arithmetic functions.  Here $i$ and $j$ are the lower and upper bounds of the search, $x$ is the array being searched, and $y$ is the element being looked for.  The procedure returns a pair $\langle b,k\rangle$ where $b$ is $true$ if the element is found, $false$ otherwise, and $k$ gives the location of the element if it is found:
 
\begin{tabular}{l}
$bsearch(i,j,x,y) \rightarrow$\\
\ \ \ if $i > j$ then $\langle false, \_\_ \rangle$ else\\
\ \ \ \ \ \ if $eq(	i,j)$ then (if $eq(arg(i,x),y)$ then $\langle true, i\rangle$ else $\langle false, \_\_ \rangle$) else\\
\ \ \ \ \ \ if $y < arg(\lfloor (i+)j/2 \rfloor, x)$ then $bsearch(i,\lfloor (i+j)/2 \rfloor,x,y)$ else $bsearch(\lfloor (i+j)/2 \rfloor, j, x, y)$\\
\end{tabular}

For this procedure $p$, $E_p$ is the text on the right-hand side of the rule.
One could also do an insertion sort using $replace$ and $arg$.

The append function on lists can be written as follows, using the old LISP notation for lists in which $cons(x,y)$ is the list $y$ with $x$ added to the front, and $NIL$ is the empty list:

\begin{tabular}{l}
$append(cons(u,v),w) \rightarrow cons(u, append(v,w))$\\
$append(NIL, w) \rightarrow w$\\
\end{tabular}

For the append procedure, the only constructors needed are $cons$ and $NIL$.  If there are other constructors, then $append$ would also have to be given a definition on them, to achieve completeness.  If sorted term rewriting were used, then it would be sufficient to specify that the arguments of $append$ must be of the appropriate sort.

Following is a  procedure that computes the length of a list:

\begin{tabular}{l}
$length(cons(u,v)) \rightarrow 1 + length(v)$\\
$length(NIL) \rightarrow 0$\\
\end{tabular}

This (flat) procedure zeroes out a range of elements in an array:

\begin{tabular}{l}
$zeroint(i,j,x) \rightarrow$ if $i > j$ then $x$ else $zeroint(i+1,j,replace(i,0,x))$\\
\end{tabular}

Because of the tail recursion, this procedure could be translated into an iterative program in an imperative programming language.

It should be clear how one could write programs using this formalism.  Note that the correctness of such programs follows in a sense from the declarative semantics of the program, assuming that all of the equations are correct.

\section{Equivalence Relation on Terms}

In an implementation of rewriting, it is convenient to store all occurrences of the same subterm only once and have pointers to this subterm.  This technique is already well known for functional programming and term rewriting.  This idea also impacts translations of the rewrite programs to imperative programming languages.  This identification of repeated subterms can be formalized by an {\em equivalence relation} on subterm occurrences in a program snapshot.  The intention is that equivalent subterm occurrences would all be stored in the same location.  This treatment of repeated subterms will also influence destructive operations, which are necessary for efficient implementations.

For this, the equivalence classes are named using {\em identifiers}.  A {\em decorated function symbol} is a pair $id:f$ where $f$ is an ordinary {\em plain} function symbol and $id$ is the name of an equivalence class.  A {\em decorated term} is a term in which all the function symbols are decorated function symbols.  If $id:f(t_1, \dots, t_n)$ is a decorated term, then it is in the $id$ equivalence class.  The relation $s \sim t$ on decorated terms
is defined so that $id:f(s_1, \dots, s_m) \sim id':g(t_1, \dots, t_n)$ only if $f = g$, $m = n$,
$id = id'$, and $s_i \sim t_i$ for all $i$, $1 \le i \le n$.  If $id:f(s_1, \dots, s_m)$ and $id':g(t_1, \dots, t_n)$ are two decorated terms and $id = id'$ then it must be true that
$id:f(s_1, \dots, s_m) \sim id':g(t_1, \dots, t_n)$. The identifier $id$ is called an
{\em equivalence label} because there may be identical terms that have different labels on their top-level symbols.  If $s$ is a decorated term then $s \rhd t$ indicates that the {\em plain} term $t$ is $s$ with all the equivalence class identifiers removed.

\subsection{Decorated Rewriting}

The rewrite relation can be extended to a {\em decorated rewrite relation} $\Rightarrow^d$ on decorated terms.  For this, a {\em decorated substitution} is of the form
$\{t_1/x_1, \dots, t_n/x_n\}$ where the $t_i$ are decorated terms.  Suppose $t[u]$ is a decorated ground term with $u$ as a subterm.  Suppose $u \rhd u'$ and $u'$ unifies with the left hand side $r$ of a rule $r \rightarrow s$ in $R$.  Thus there is a substitution $\Theta$ such that $u' \equiv r\Theta$.  Let $r^d$ and $\Theta^d$ be a decorated term and a decorated substitution such that $u \equiv r^d \Theta^d$. Let $s^d$ be a decorated $s$ where the identifiers in $s$ are chosen so that only syntactically identical subterms of $t[s^d \Theta^d]$ are in the same equivalence class.  Then $t[u] \Rightarrow^d_R t[s^d\Theta^d]$.  This implies that if the right-hand side of a rewrite rule has repeated variable occurrences, then all occurrences of a term replacing a given repeated variable will be equivalent.

The above description applies if there is only one term in the equivalence class of $u$.  If there is more than one such term, then they all have to be rewritten together.   Indicating all the occurrences of $u$ in $t$ by $t[u, u, \dots, u]$, $t[u, u, \dots, u] \Rightarrow^d_R t[s^d\Theta^d, s^d\Theta^d, \dots, s^d\Theta^d]$.  This kind of rewriting, when all identical redexes are rewritten at the same time, is called {\em parallel rewriting}.

Equivalence class names also have to be assigned to subterms of the input term (the term
given to evaluate at the start).  That is, the input term has to be decorated.  This can be done arbitrarily subject to the rule that only syntactically identical subterms can be in the same equivalence class.

\section{Assignment Statements}

In the rewriting formalism, it is convenient to allow a single assignment style of programming on the right-hand sides of rewrite rules.   Such a construction is essentially identical to the (let $x = t$ in $E$) construction in functional programming. This is similar to the imperative style of programming. For example, consider this procedure given earlier:

\begin{tabular}{l}
$zeroint(i,j,x) \rightarrow$ if $i > j$ then $x$ else $zeroint(i+1,j,replace(i,0,x))$\\
\end{tabular}

This can also be written this way:

\begin{tabular}{l}
$zeroint(i,j,x) \rightarrow$ if $i > j$ then $x$\\
else\\
\ \ \  $k \leftarrow i+1$;\\
\ \ \  $y \leftarrow  replace(i,0,x)$;\\
\ \ \  $zeroint(k,j,y)$\\
\end{tabular}

In this procedure the assignment statements do not influence the equivalence relation, but if the variable on the left-hand side of an assignment statement occurs more than once in the remainder of the procedure, then all these occurrences will be replaced by equivalent terms. An example of this is the term $x \leftarrow g(c); f(x,x)$ which is equivalent to $f(g(c),g(c))$ with the two occurrences of $g(c)$ equivalent.  So assignment statements can influence the equivalence relation and thus the decorated rewriting relation.  Other than this, the procedure with assignment statements can be considered as equivalent to the version without them.

Also, $\langle x_1, \dots, x_n\rangle$   $\leftarrow t$ is an allowed assignment statement. The intention is that $t$ evaluates (rewrites) to a term of form $\langle t_1, \dots, t_m \rangle$ for $m > n$,  and then this assignment can be regarded as the sequence $x \leftarrow t; x_1  \leftarrow arg(1,x); \dots; x_n  \leftarrow arg(n,x)$ where $x$ is a new variable.  Also $arg(i,x) = \pi_i(x)$.

Another allowable form for a term is illustrated by the term $f((x \leftarrow g(c); h(x)), (y \leftarrow h(d); g(y)))$ which is equivalent to $f(h(g(c)),g(h(d)))$.

No assignment statements are allowed on the right hand side of an assignment statement.  This rule could be relaxed, but it simplifies the formalism.

In general, let an ``aterm" be a term that may contain assignment statements.  Basically an ``aterm" is a sequence of assignment statements followed by a term.  The term, with variables replaced as specified by the assignment statements, is the value of the ``aterm." Then we have the syntax of rewrite rules with assignment statements as follows:

$\langle$ assignment $\rangle$  := $\langle$ variable $\rangle$   $\leftarrow$ $\langle$ term $\rangle$ 

$\langle$ assignment $\rangle$  := $\langle$ left tuple delimiter $\rangle$  $\langle$ variable list $\rangle$  $\langle$ right tuple delimiter $\rangle$    $\leftarrow$ $\langle$ term $\rangle$ 

$\langle$ rewrite rule $\rangle$  := $\langle$ term $\rangle$  $\rightarrow$ $\langle$ aterm $\rangle$ 

$\langle$ aterm $\rangle$  := $\langle$ term $\rangle$ 

$\langle$ aterm $\rangle$  := $\langle$ assignment $\rangle$  ; $\langle$ aterm $\rangle$

$\langle$ aterm $\rangle$  := (if $\langle$ term $\rangle$ then $\langle$ aterm $\rangle$ else $\langle$ aterm $\rangle$)

Also, the sequence $x_1 \leftarrow t_1 ; \dots ; x_n  \leftarrow t_n ; E$ is regarded as associating to the right

The statements ($x \leftarrow t; E[x]$) are considered as a representation for $E[t]$ with the understanding that all occurrences of $x$ are replaced by $t$, but this has to be defined
more carefully later.

A statement of the form (if $B$ then $x$ else $y$); $E$ is considered as an abbreviation for (if $B$ then $x;E$ else $y;E$) assuming that
$x ; E$ and $y ; E$ are $\langle$ aterm $\rangle$.  Thus for example $x$ and $y$ can be sequences of assignment statements.

\subsection{Variable bindings}

There are restrictions on where variables may appear in terms having assignment statements.  These restrictions are defined in terms of binding or scoping of variables.

In an expression of the form $x \leftarrow t; E[x]$, all occurrences of $x$ in $E$ are {\em bound} by the assignment statement $x \leftarrow t$.  However, occurrences of $x$ in $t$ are not bound by this assignment statement.

A variable occurrence that is bound by some assignment statement is said to be {\em bound}.  A variable that is not bound is {\em free}.

The binding rule for variable occurrences is this: In an assignment statement $x \leftarrow s$, if any variable $y$ appears in $s$ then all its occurrences in $s$ must be bound, except that any  occurrences of a variable $x$ on the left hand side of an assignment statement must be free.

These restrictions imply that in a sequence $x_1 \leftarrow t_1 ; \dots ; x_n \leftarrow t_n ; E$, $x_i$ cannot appear in $x_j  \leftarrow t_j$ for $j  <  i$.  All the possibilities for $x_i$ being bound or free in $x_j \leftarrow t_j$ end up violating some restriction on binding.  These restrictions also imply that in an assignment statement $x \leftarrow t$, $x$ cannot appear in $t$.  This prohibits assignment statements such as $i \leftarrow i+1$.

\subsection{Order of eliminating assignment statements}

Because the statements ($x \leftarrow t; E[x]$) are considered as a representation for $E[t]$, where in $E[x]$, $x$ refers to free occurrences of $x$ in $E$, assignment statements can be eliminated by replacing ($x \leftarrow t; E[x]$) by $E[t]$.  If there is  more than one assignment statement, then the question arises whether the final result depends on the order in which the assignment statements are eliminated.  It turns out that the result of eliminating assignment statements in a rewrite rule does not depend on the order in which they are eliminated, and even the equivalence class names are not affected.  Thus a rewrite rule containing assignment statements unambiguously represents one without them.   In this way one can speak of {\em the} assignment free form of a program.  For this, the terms can be considered as decorated terms to make it clear that the final result, including the names of the equivalence classes, does not depend on the order of elimination.  Also, the process of elimination must terminate because the number of assignment statements in $E[t]$ is one less than the number in ($x \leftarrow t; E[x]$) because $t$ does not contain any assignment statements.  The following result was shown in \cite{BarnettPlaisted2018}:

\begin{theorem}
The result of eliminating assignment statements from a rewrite procedure does not depend on the order in which the statements are eliminated.
\end{theorem}

\subsection{Multiple Assignments into a Variable}

In order to make the language closer to imperative languages, it is convenient to allow multiple assignments into the same variable. This requires a modification of the scoping rules
above, so that in an assignment statement $x \leftarrow t$, $x$ can be bound.  In such a case, it is convenient to specify the order in which the assignment statements are evaluated. 
This even becomes necessary if there are destructive operations, as we shall see.  Therefore it is necessary to specify an evaluation strategy. Consider an expression of the form $x \leftarrow t_1 ; x \leftarrow t_2 [x]; E[x]$.    This can be regarded as equivalent to
$x_1 \leftarrow t_1; x_2 \leftarrow t_2[x_1]; E[x_2]$.   However, there is another method to
specify the meaning of such assignments that
avoids the need to introduce new variables, and is more elegant.

\subsection{Eliminating assignment statements}

Letting $[[E]]$ be the result of eliminating
assignment statements from an expression $E$, the meaning of multiple assignment statements into
a variable can be expressed more elegantly by the rules

\[[[x \leftarrow t ; E]] = [[E]] (t/x)\]
\[[[\langle x_1, \dots, x_n \rangle \leftarrow t ; E]] = [[x \leftarrow t ; x_1 \leftarrow \pi_1(x) ; \dots, x_n \leftarrow \pi_n(x) ; E]]\]

The first rule works because $[[E]]$ will be a term without any assignment statements.  Essentially the assignment statements are eliminated innermost first.  Also, if $f$ is a defined or constructor term (other than semicolon) then 

\[[[f(t_1, \dots, t_n)]] = f([[t_1]], \dots [[t_n]]).\]

Finally, a procedure
definition $p(x_1, \dots, x_n) \rightarrow E$ is converted into the rewrite rule

\[p(x_1, \dots x_n) \rightarrow [[E]]\]

and $R_p$ is the singleton set containing this rule.

One can show that $[[E]]$ will contain no assignment statements for an expression $E$ of the form
$\langle aterm \rangle$ or $\langle term \rangle$.  Variables on the left-hand sides of assignment statements will be eliminated from the term
rewriting system associated with a procedure when the assignment statements are eliminated, so
these variables do not need to be considered as part of the
term rewriting system.

\section{Destructive Operations }

Destructive operations are needed for efficiency in imperative programs, but do not conform to pure term rewriting semantics.  Therefore it is important to have a way to relate the two semantics.

We illustrate the problem with an assignment to an element of an array. This example uses {\em decorated rewriting}, so all equivalent subterms are rewritten at once.

Suppose the operation $replace(i,t,A(s_1, \dots , s_n))$ is done where $A(s_1, \dots, s_n)$ represents a one dimensional array with $n$ elements.  This operation replaces the $i$-th argument of $A$ with $t$.  In rewriting semantics, the result is $A(s_1, \dots, t, \dots, s_n)$, where the term $t$ replaces $s_i$.  This operation creates (copies) a completely new term, and other occurrences of $A(s_1, \dots, s_n)$ are not affected.  In $d$-$replace(i,t,A(s_1, \dots , s_n))$, the same storage is used for $A(s_1, \dots, t, \dots,  s_n)$ as was used for $A(s_1, \dots, s_n)$,and just the $i$-th element is modified, so all other occurrences of $A(s_1, \dots, s_n)$ in the same equivalence class are also modified to $A(s_1, \dots, t, \dots,  s_n)$.  Thus the destructive operation may have a {\em side effect}.
Such destructive operations are necessary for efficient imperative programs, but do not conform to the pure term rewriting semantics if there are any other occurrences of $A(s_1, \dots, s_n)$ in the same equivalence class.  A similar situation can occur for a change to a pointer in the middle of a list; other references to this list will also have the pointer changed if destructive operations are used.

To make this more precise, the notation $t[u_1, \dots, u_n] \Rightarrow t[v_1, \dots, v_n]$ can be used to illustrate that all the changes of the terms $u_i$ to $v_i$ occur together. We also refer a term $u$ in equivalence class $id$ by $id:u$ as before.  So for the array modification example we have
\begin{gather*}
t[id':d\_replace(i,t,id:A( s_1 , \dots,  s_n)), id:A( s_1 , \dots,  s_n)] \Rightarrow\\
t[id:A( s_1  , \dots,  t , \dots,  s_n), id:A( s_1 , \dots,  t , \dots,  s_n)]\\
t[id':replace(i,t, id:A( s_1 , \dots,  s_n)), id:A( s_1 , \dots,  s_n)] \Rightarrow\\
t[id':A( s_ 1 , \dots,  t , \dots,  s_n), id:A( s_1 , \dots,  s_n)]
\end{gather*}
In the expression $t[ \dots ]$, the terms inside the brackets indicate distinct subterm occurrences in a term $t$.
In general, for an operation $f(s_1, \dots, s_n)$ that returns the $i$-th argument $s_i$ but destructively modifies it to $t$, we have $r[f( s_1 , \dots,  s_n), s_i, u[s_i]]$ $\Rightarrow r[t, t, u[t]]$ but for non-destructive operations $r[f ( s_1 , \dots,  s_n), s_i, u[s_i]] \Rightarrow r[ t,s_i, u[s_i]]$.

This problem can be alleviated for destructive operations by an explicit copy operation before the destructive operation; the copy operation creates a new syntactically identical term but not in the same equivalence class.  Then the destructive operation is done on the copy.  However, this has a cost in efficiency.

Now, programs with destructive operations may not be confluent.  Here is an example:

\begin{example}
\label{non.confluence.example}

Starting term $f(\langle1,2\rangle)$

$f(x) \rightarrow \langle arg(1,x),d\_replace(1,2,x) \rangle$

$arg(i,x)$ gives the $i$-th argument of $x$

$d\_replace(1,2,x)$ replaces the first argument with 2 in $x$ destructively
\end{example}

If the first element of the tuple is evaluated first we get $\langle1,\langle 2,2\rangle \rangle$, which agrees with the semantics for non-destructive replace, else we get $\langle 2,\langle 2,2 \rangle \rangle$, which does not, noting that the two occurrences of $x$ on the right hand side of the rule are equivalent (in the same equivalence class).  The point of this is not to show that orthogonal
rewriting with common subterms is not confluent, but to illustrate the problems that come with giving a rewriting semantics and with destructive operations.

Because programs with destructive operations are not confluent, an evaluation strategy has to be specified to give a program with destructive operations a precise meaning.  The body of a procedure
definition can be replaced by a term, but the order of evaluation of the subterms has to be retained in
order to specify the effects of destructive operations.

It is possible for destructive operations to produce a circular structure, in which a term
has itself as a proper subterm.  For simplicity we assume that this does not happen, although it does not cause significant problems if it does.

\section{Iterative Statements}

Imperative languages generally have iterative statements such as $for$, $while$, and $until$.  It would be good to have these also in this language.  Here these statements are defined by recursive
calls, but the intention is that these statements would translate into corresponding iterative statements in the target imperative languages, and not recursive calls.  To define iterative statements, the procedures $p_{for}$, $p_{while}$, and $p_{until}$ are used, with definitions as follows:
\begin{gather*}
p_{for, A, E, i}(i_0,n,\overline{x}) \rightarrow  i \leftarrow i_0; \overline{y} \leftarrow \overline{x};\mbox{ if } i \le  n \mbox{ then } (A ; p_{for,A,E,i}(i_0+1,n,\overline{y})) \mbox { else } E\\
p_{while,B,A,E}(\overline{x}) \rightarrow \overline{y} \leftarrow \overline{x}; \mbox{ if } \neg B \mbox{ then } E \mbox{ else } (A ; p_{while,B,A,E}(\overline{y}))\\
p_{until,B,A,E}(\overline{x}) \rightarrow \overline{y} \leftarrow \overline{x}; A; \mbox{ if } B \mbox{ then } E \mbox{ else } (A ; p_{until,B,A,E}(\overline{y})\\
\end{gather*}
The variables $\overline{x}$ are the variables whose values $A$ or $E$ needs, whose value is computed outside of these segments.  The assignment $\overline{y} \leftarrow \overline{x}$ copies these values into the appropriate variables.  Sometimes these variables will be omitted.
$p_{for,A[i],E[i],i}(i_0,n,\overline{x})$ is similar to the construct $\overline{y} \leftarrow \overline{x}$; (for $i = i_0$ step 1 until $n$ do $A[i]) ; E[i]$.  $p_{while,B,A,E}(\overline{x})$ is similar to the construct $\overline{y} \leftarrow \overline{x}$; (while $B$ do $A$); $E$.  $p_{until,B,A,E}(\overline{x})$ is similar to the construct
$\overline{y} \leftarrow \overline{x}$; (do $A$ until $B$); $E$.  In a system implementing these ideas, the user would write programs using the usual for, while, and until syntax, and these would be translated into $p_{for}$, $p_{while}$, and $p_{until}$ to obtain a term-rewriting system for purposes of proving properties or translation into imperative programming languages.  For nested
iterative statements, the outermost iterative statement can be replaced by a recursive procedure call, and then the same thing can be done for the next innermost iterative statement, and so on.

To illustrate these constructs, consider this example program:

\begin{tabular}{l}
$p(i_0,n,x_1) \leftarrow$\\
\ \ $y_1 \leftarrow x_1$;\\
\ \ for $i = i_0$ step 1 until $n$ do\\
\ \ \ \ $y_1 \leftarrow 2 * y_1$;\\
\ \ \ \ $y_1 \leftarrow y_1 + 1$;\\
\ \ $y_1$;\\
\end{tabular}

This program is equivalent to $p_{for,A,E,i}(i_0,n,x_1)$ where $A = (y_1 \leftarrow 2 * y_1 ; y_1 \leftarrow y_1+1)$ and $E = y_1$.  In the recursive notation, this can be written as follows.

\begin{tabular}{l}
$p(i_0,n,x_1) \leftarrow$\\
\ \ $i \leftarrow i_0$;\\
\ \ $y_1 \leftarrow x_1$;\\
\ \ if $i \le n$ then $y_1 \leftarrow 2 * y_1$;\\
\ \ \ \ $y_1 \leftarrow y_1 + 1$;\\
\ \ \ \ $p(i_0+1,n,y_1)$\\
\ \ else $y_1$;\\
\end{tabular}

Eliminating assignment statements from the recursive program,
\[[[y_1 \leftarrow 2 * y_1; y_1 \leftarrow y_1 + 1; p(i_0+1,n,y_1)]]=p(i_0+1,n,2*y_1+1)\]
so the entire program becomes
\[\mbox{if } i_0 \le n \mbox{ then } p(i_0+1,n,2*x_1+1) \mbox{ else } 2*x_1 + 1\]
and the rewrite rule for this program is
\[p(i_0,n,x_1) \rightarrow \mbox{if } i_0 \le n \mbox{ then } p(i_0+1,n,2*x_1+1) \mbox{ else } 2*x_1 + 1\]

Rewriting is done leftmost innermost except for conditional statements (if $B$ then $x_1$ else $x_2$) in which the condition is evaluated and then one of the branches $x_1$ or $x_2$.  Consider the evaluation of $p(1,3,1)$. This rewrites to $\mbox{if } 1 \le 3 \mbox{ then } p(1+1,3,2*1+1) \mbox{ else } 2*1 + 1$ which by compiled functions and evaluating the conditional
yields $p(2,3,3)$. Rewriting this yields $\mbox{if } 2 \le 3 \mbox{ then } p(3,3,2*3+1) \mbox{ else } 2*3 + 1$ which yields $p(3,3,7)$. Rewriting this yields $\mbox{if } 3 \le 3 \mbox{ then } p(3+1,3,2*7+1) \mbox{ else } 2*7 + 1$ which yields $p(4,3,15)$. Rewriting this and evaluating as before yields 15.

\section{Defining algorithms}

Pseudocode is used in many algorithm texts to define algorithms.  However, this code generally
does not have a precisely specified syntax and semantics.  Here is what a well-known algorithms text \cite{Cormen:2009:IAT:1614191} says on page 16:

\begin{quotation}
In this book, we shall typically describe algorithms as programs written in a {\em pseudocode} that is similar in many respects to C, C++, Java, Python, or Pascal.
\end{quotation}

Another problem is that there is not a precise definition for what an algorithm is \cite{Vardi:2012:ALG:2093548.2093549,DBLP:conf/sofsem/2012,Hill:2015,DBLP:journals/bsl/BlassDG09,Moschovakis2001,journals/eatcs/BlassG03,Gurevich:2000:SAM:343369.343384,Blass03abstractstate}.  The first algorithm in the text by Cormen et al \cite{Cormen:2009:IAT:1614191} is insertion sort.  However, it is not precisely defined how one tells if an algorithm is insertion sort or not.

While not claiming to solve this problem, the rewriting language can be used as pseudocode to define algorithms.  It does have a precise syntax 
and semantics, so that in principle programs can be verified in it.  It also has the appearance of imperative
languages, making it suitable for defining algorithms in an imperative style.  Furthermore, translations of this language into high-level programming languages can be written in a way that preserves the
structure of the algorithm, in that assignment statements translate to assignment statements, iterative
statements to iterative statements, procedure calls to procedure calls, and conditional statements to
conditional statements.

Here is what the insertion sort looks like in the algorithms text referenced above \cite{Cormen:2009:IAT:1614191}:

Insertion Sort(A)

\begin{tabular}{ll}
1 & \ \ for $j = 2$ to $A.length$\\
2 & \ \ \ \ $key = A[j]$\\
3 & \ \ \ \ Comment\\
4 & \ \ \ \ $i = j-1$ \\
5 & \ \ \ \ while $i > 0$ and $A[i] > key$\\
6 & \ \ \ \ \ \ $A[i+1] = A[i]$\\
7 & \ \ \ \ \ \ $i = i-1$\\
8 & \ \ \ \ $A[i+1] = key$
\end{tabular}

In the rewriting language this could be defined as follows, letting $A[j]$ be a notation for
$arg(j,A)$ and $A[i] \leftarrow E$ be a notation for $A \leftarrow$ $d\_replace$($i, E, A$):

\begin{tabular}{ll}
0 & $InsertionSort(A, length) \leftarrow$\\
1 & \ \  for $j = 2$ step 1 until $length$ do\\
2 & \ \ \ \ $key \leftarrow A[j]$;\\
3 & \ \ \ \ Comment;\\
4 & \ \ \ \ $i \leftarrow j-1$; \\
5 & \ \ \ \ while $i > 0$ and $A[i] > key$ do\\
6 & \ \ \ \ \ \ $A[i+1] \leftarrow A[i]$;\\
7 & \ \ \ \ \ \ $i \leftarrow i-1$;\\
\end{tabular}

\begin{tabular}{ll}
8 & \ \ \ \ $A[i+1] \leftarrow key$;\\
9 & $A$\\
\end{tabular}

Statement 9 is needed for the purpose of returning the array $A$ as the value of the procedure.  Iterative statements must be followed by an expression returning a value in the rewriting language.  When the
outermost iterative statement is replaced by a procedure call, then the inner statement will be
followed by such an expression, so it does not need to explicitly return a value.

\section{Order of evaluation}
Because the effect of destructive operations can depend on the order of evaluation, it is
necessary to fully specify the order of evaluation. The order of evaluation specified here seems to correspond to that of typical
imperative languages; right-hand sides of assignment statements are executed in order, and conditional statements are
executed in a typical way.  The purpose is not to define the most efficient execution of term-rewriting but to provide a precise and
simple rewriting-based
semantics for algorithms written in a typical imperative style.

Recall that
$R_P$ is the term rewriting system for a program $P$.  Additional rewrite rules correspond to
the effect of destructive operations; with these added one obtains a system $R_P^+$.  The {\em starting query} is
the term $s$ given at the start, and it is desired to reduce it to normal form using
$R_P^+$.  For this, a procedure $eval(t,T)$ is defined at the meta-level.  Here
$T$ is a term such that $s \Rightarrow_R^* T$ (dropping the subscript from $R$) and
$t$ is a subterm of $T$.  Assuming that it terminates, $eval$ computes a normal form $t'$ of $t$ and a term $T'$
obtained by replacing all occurrences of $t$ in $T$ by $t'$.  At the top level $eval(s,s)$
is called to compute a normal form of $s$.  If $\beta$ is an equivalence label then
$T|(\beta \leftarrow u)$ is term $T$ with all occurrences of subterms having label
$\beta$ replaced by $u$.

To define $eval$, consider the body $E_p$ of
a flat procedure $p$.  Let $\alpha$ be a substitution replacing formal parameters by
actual parameters.  Then let $E_p\alpha/T$ be $E_p\alpha$
with the equivalence labels of terms and subterms $u$ on right-hand sides of assignment statements chosen to be distinct from any other equivalence labels in $T$, except that
equivalence labels of constructor terms need not be distinct from labels of identical terms
in $T$.  Also, subterms that are variables do not need to be assigned an equivalence
label.

Define $eval\_list$ as follows:

\begin{tabular}{l}
$eval\_list((t_1, \cdots, t_n),T) \rightarrow$\\
\ \ \ \ $(t'_1,T_1) \leftarrow eval(t_1,T)$;\\
\ \ \ \ $(t'_2,T_2) \leftarrow eval(t_2,T_1)$;\\
...\\
\ \ \ \ $(t'_n,T_n) \leftarrow eval(t_n,T_{n-1})$;\\
$((t'_1, \cdots, t'_n),T_n)$\\
\end{tabular}

If $p$ is a flat procedure with body $E_p$ then $eval$ is defined as follows:

\begin{tabular}{l}
$eval(\beta:p(t_1, \cdots, t_n),T) \rightarrow$\\
\ \ \ \ $((t'_1, \cdots, t'_n), T_n) \leftarrow eval\_list((t_1, \cdots, t_n),T)$;\\
\ \ \ \ $\alpha = (x_1 \leftarrow t'_1, \cdots, x_n \leftarrow t'_n)$;\\
\ \ \ \ $eval(E_p\alpha/T_n,T_n|(\beta \leftarrow [[E_p\alpha/T_n]]))$;\\
\end{tabular}

Roughly speaking, $[[E_p\alpha/T_n]]$ is the body of $p$ converted to a term by
eliminating assignment statements, and
subterms of this term will be evaluated in an order corresponding to the order of the statements in
$E_p$.  
For assignment statements, $eval$ is defined as follows:

\begin{tabular}{l}
$eval((x \leftarrow t; E),T) \rightarrow$\\
\ \ \ \ $(t',T') \leftarrow eval(t,T)$;\\
\ \ \ \ $eval(E,T')$;\\
\end{tabular}

Roughly speaking, the subterms of $T$ corresponding to $t$ are evaluated, and then
the remaining terms in $E$ are evaluated.  For conditional statements, $eval$ is defined
as follows:

\begin{tabular}{l}
$eval$($\beta$:(if $B$ then $E_1$ else $E_2$),$T$) $\rightarrow$\\
\ \ \ \ $(B',T') \leftarrow eval(B,T)$;\\
\ \ \ \ if $B' =$ true then $eval(E_1,T'|(\beta \leftarrow E_1))$\\
\ \ \ \ else if $B' = $ false then $eval(E_2,T'|(\beta \leftarrow E_2))$\\
\ \ \ \ else error;\\
\end{tabular}

For constructors $eval$ is defined as follows:

\begin{tabular}{l}
$eval$($f(t_1, \cdots, t_n)$,$T$) $\rightarrow$\\
\ \ \ \ if $f$ is a constructor and the $t_i$ are constructor terms then $(f(t_1, \cdots, t_n),T)$ else\\
\ \ \ \ if $f$ is a constructor then\\
\ \ \ \ \ \ $((t'_1, \cdots, t'_n),T_n) \leftarrow eval\_list((t_1,\cdots, t_n),T)$;\\
\ \ \ \ \ \ $(f(t'_1, \cdots, t'_n),T_n)$\\
\end{tabular}

If $p$ is a rewrite procedure then $eval$ is defined as follows:

\begin{tabular}{l}
$eval(\beta: p(t_1, \cdots, t_n),T) \rightarrow$\\
\ \ \ \ Find $t'$ such that $p(t_1, \cdots, t_n) \Rightarrow_{R_p} t'$ by the evaluation
strategy specified for $p$;\\
\ \ \ \ $eval(t',T|(\beta \leftarrow t'))$\\
\end{tabular}

Compiled procedures can be considered as rewrite procedures with infinitely many rewrite rules, as mentioned earlier.  For destructive replacement $eval$ is defined as follows:

\begin{tabular}{l}
$eval(\alpha:d\_replace(i,u,t),T) \rightarrow$\\
\ \ \ \ $((i',u',\beta:t'),T') \leftarrow eval\_list((i,u,t),T)$;\\
\ \ \ \ Suppose $t'$ is $f(s_1, \cdots, s_n)$ and $n \ge i'$;\\
\ \ \ \ Let $w$ be $\beta:f(s_1, \cdots, s_{i'-1}, u', s_{i'+1}, \cdots, s_n)$;\\
\ \ \ \ $(w,T'|(\beta \leftarrow w)|(\alpha \leftarrow w))$\\
\end{tabular}

First the rewrite rule $t' \rightarrow w$ is applied to $T'$; this rule rewrites a constructor
term into a different constructor term and thus is not a logical consequence of $R_P$.  This
rule rewrites all occurrences of $t'$, which can introduce side effects.  However, if there
are no other references to $t'$ except in this call to $d\_replace$, the rewrite semantics
is not violated.  Then the $d\_replace$ call is replaced by $w$ which is its value.

This semantics for destructive operations permits actual parameters to be modified during the execution of a procedure.  For example,
one could pass an array to a procedure that zeroes all its elements.

\section{Conclusion}
A language based on first-order term rewriting has been developed to have something of the style of imperative programming languages and to permit efficient translations into such languages.  This language also permits destructive operations with side effects.  To model efficient imperative languages, an equivalence relation on term occurrences can be used, and a version of rewriting called decorated rewriting is developed for this purpose.  The pure rewriting language is confluent, but the version with destructive operations is not.   The approach given in this paper could be extended to more sophisticated languages such as Coq with a highly developed type system.

We wish to emphasize that this language has been made as simple as possible to meet the stated objective of providing an efficient and direct translation into imperative programming languages.  Therefore extensions such as higher-order unification and AC-unification have not been added.  This simplicity makes it easier for the typical programmer to understand the language and also makes the language easier to translate, interpret, and possibly compile.  It also facilitates proofs of correctness in the language.

We also have tried to make the language as close as possible to typical imperative programming languages.  It is our intention that assignment statements in this formalism
will translate to assignment statements in an imperative language, recursion will translate
to recursion, and iterative statements will translate to iterative statements.  This means that the translated program will have essentially the same structure as the original program.  This gives the programmer a high degree of control over the form of the translated program and also makes it easier to translate the proof of correctness and to maintain and perform complexity analysis on the translated programs.  This also makes the language attractive as a replacement for the pseudo-code often used in typical algorithms texts to specify algorithms, but without a rigorous abstract semantics.

In the current programming environment, the same algorithms are programmed again and again in different programming languages, but the structure of the code is essentially the same.  It would be better to represent the algorithm once in an abstract language such as that presented here and then write translations into various high-level imperative languages.  This also makes it more economical to verify the algorithm once in the abstract setting rather than verifying each implementation of it in a specific language; the algorithm
will typically only be part of a program, and the entire program will still be verified in the
target language.
We also want to emphasize that our concern is with algorithms such as shortest path algorithms, maximum flow algorithms, and others that may be used again and again.  We are not interested in programs that use many interrupts, for example, or even in entire programs, but only in portions of them that can be regarded as algorithms.  The present approach can be used to provide implementations of algorithms in larger programs in which some of the code is provided directly in the target language.

In order to avoid conflicts with procedures in the target language program, there should be a facility for renaming the defined symbols and constructors in the abstract programs.  Also, it would be good to have encapsulation to help with this.  However, a detailed translation example is beyond the scope of this paper.

Another problem with the current approach to programming is that programs generally become obsolete after about 15 years and have to be updated or completely rewritten. It should be possible to write an algorithm once and never have to write it again.  The system presented in this paper is designed to make this possible and is maximally accessible to the typical programmer.

\bibliography{paper2}

\end{document}